\documentclass[11pt, arxiv]{lapreprint-mod}
\usepackage[version=4]{mhchem} % For chemical notation
\usepackage{siunitx}    % For SI units
\usepackage{pdflscape}  % For putting pages in landscape mode
\usepackage{rotating}   % For rotating specific elements
\usepackage{textcomp,gensymb}    % Symbols
\usepackage[misc]{ifsym} % For the \Letter symbol
\usepackage{orcidlink}  % For the \orcidlink
\usepackage{listings}   % For inserting code chunks
\usepackage{colortbl}   % For Knitr table colouring
\usepackage{tabularx}   % For making Knitr tables compatible
\usepackage{longtable}  % For multi-page tables
\usepackage{subcaption}
\usepackage{multirow}
\usepackage{snotez}     % For sidenote environments. enotez for endnotes
\usepackage{csquotes}   % For language-based quote rules (helps BiBLaTeX)
\usepackage{annotate-equations}
\usepackage{amsmath}
\usepackage{amsfonts}
\usepackage{amssymb,amsthm,mathtools}
\newtheorem{definition}{Definition}
\newtheorem{theorem}{Theorem}
\newtheorem*{theorem*}{Theorem}

\newtheorem{proposition}{Proposition}
\newtheorem{remark}{Remark}
\newtheorem{lemma}{Lemma}
\newtheorem{claim}{Claim}
\newtheorem{example}{Illustrative example}

\usepackage[			% use biblatex for bibliography
	backend=biber,      % use biber or bibtex backend
   style=authoryear,   % choose style
	natbib=true,		% allow natbib commands
	hyperref=true,	    % activate hyperref support
	alldates=year,      % only show year (not month)
   uniquename=false,   % don't add firstnames when citing multiple sources by the same author
   maxbibnames=99,     % maximum number of author names to list in bibliography before 'et al' is used instead
]{biblatex}

\AtEveryBibitem{
   \clearfield{urlyear}
   \clearfield{urlmonth}
   \clearlist{language}
}

\title{Two is enough: a flip on Bertrand through positive network effects}

\author[1 \Letter]{Renato Soeiro}
\author[1]{Alberto A. Pinto}

\affil[1]{INESC TEC, Faculdade de Engenharia, Universidade do Porto, Portugal}

\metadata[]{\Letter\\ For correspondence}{rsoeiro@fe.up.pt}
\leadauthor{Soeiro \& Pinto}
\shorttitle{Two is enough: a flip on Bertrand competition}

\begin{document}
\maketitle
\begin{abstract}
    We discuss price competition when positive network effects are the only other factor in consumption choices. We show that partitioning consumers into two groups creates a rich enough interaction structure to induce negative marginal demand and produce pure price equilibria where both firms profit. The crucial condition is one group has centripetal influence while the other has centrifugal influence. The result is contrary to when positive network effects depend on a single aggregate variable and challenges the prevalent assumption that demand must be micro-founded on a distribution of consumer characteristics with specific properties, highlighting the importance of interaction structures in shaping market outcomes. 
\end{abstract}
\section{Introduction}
The enduring challenge of price competition models is dealing with the price undercutting dynamics' tendency to generate extreme outcomes. The enduring influence of \citealt{Bertrand1883}'s argument, which turned into the widespread textbook model on duopoly price competition, underscores the importance and difficulty of the challenge. Positive network effects amplify the problem, as the possibility of a bandwagon effect creates a further incentive for price deviations that potentially allow capturing the whole market. Sufficient heterogeneity is the usual answer to resolve the disparity between the extreme equilibrium outcomes produced by simple price competition models and the stability of shared and profitable market outcomes observed in practice. Sufficient consumer heterogeneity mitigates the gains from price deviations by breaking the homogeneous response to price changes and providing negative elasticity to demand.

How much is sufficient?
In this note, we show that partitioning consumers into two groups creates a rich enough interaction structure to produce pure price equilibria where both firms realize (positive) profits, even when positive network effects are the only factor besides price. \textit{Two is enough} to induce the desired properties on the demand function.
\subsection{Motivation}
A way to microfound demand for a price competition model is to adopt a two-stage game framework. In the first stage, firms (sellers) set prices, while in the second, after knowing prices, consumers (users/buyers) choose which firm they want to buy from. The second-stage player's strategy is a contingent plan, which determines a demand function. Subgame perfection provides a foundation for demand by restricting actions to the set of Nash equilibria. However, a demand function is not uniquely determined because of the multiplicity of second-stage equilibria. A demand function is a selection from the correspondence that associates each set of prices with the corresponding set of second-stage Nash equilibria. Together with prices, the latter forms the strategy pair of the two-stage game. The problem of the existence of a (subgame-perfect) equilibrium where both firms have positive profit is the following: is there any selection that produces the desired outcome? Is there a way to microfound the outcome?

For clarity of argument, let us discuss the concepts in a classical Bertrand duopoly model. Two firms $a$ and $b$ set prices, and $N$ consumers choose to maximize $v(j)=U-p^j$, $j=a,b$, where $U\in \mathbb{R}$, hence, inducing Bertrand's argument that consumers want to buy from the cheapest firm. Considering total coverage, a consumers' choice for a given pair $\mathbf{p}=(p^a,p^b)$ can be represented by $\sigma(\mathbf{p})\in [0,1]$, the fraction of consumers that choose $a$. Hence, demand for $a$ and $b$ at prices $\mathbf{p}$ is, respectively, $D^a=\sigma N$ and $D^b=(1-\sigma) N$. The \textit{equilibrium demand correspondence} $Q: \mathbb{R}_{\ge 0}^2 \to [0,N]$ associates to each $\mathbf{p}$ the set of values $\sigma N$ such that $\sigma$ represents a second-stage (consumption) Nash equilibrium for $\mathbf{p}$. Demand functions  $D^j(\mathbf{p})$ are determined by a \textit{selection} from $Q$ for each $\mathbf{p}$, and we refer to $\sigma$ as a selection. Firms want to maximize $\Pi^j(\mathbf{p})=p^j D^j(\mathbf{p})$, $j=a,b$ (no costs for simplicity). An equilibrium is $(\mathbf{p},\sigma)$ where $\sigma$ is a selection whose restriction to $\mathbb{R}_{\ge 0}\times \{p^b\}$ is such that $\Pi^a(p^a,p^b)\ge\Pi^a(p,p^b)=p N \sigma(p,p^b)$ for all $p\in \mathbb{R}_{\ge 0}$, and similarly for $b$. The restriction $Q|_{\mathbb{R}_{\ge 0}\times \{p^b\}}$ is depicted in Figure \ref{fig:A}. Bertrand's paradox is that any selection leads to the equilibrium prices $\mathbf{p}=(0,0)$. The only difference in selections is, in fact, the \textit{tie-breaking rule}. ($Q$ is a singleton for all prices except $p^a=p^b$, for which it is $Q(\mathbf{p})=[0,N]$.)

In its simplest form, network effects are a term in $v(j)$ proportional to the number of other consumers who choose $j$, according to a parameter $\alpha>0$, that is, $v(a)=U-p^a+\alpha N \sigma$, analogously for $v(b)$. Consumers may opt for the more expensive product if the network effects offset the difference. When the price gap is small, various equilibria can emerge. The correspondence $Q$ folds: both monopolies exist beyond the diagonal $p^a=p^b$, connected by an upward-sloped line (plane) where consumers \textit{split}. The (restricted) correspondence is depicted in Figure \ref{fig:A}. In this case, it is no longer true that any selection will lead to $\mathbf{p}=(0,0)$. Some lead to monopoly equilibria.

\begin{figure}[ht]
    \begin{fullwidth}
        \includegraphics[width=0.67\textwidth]{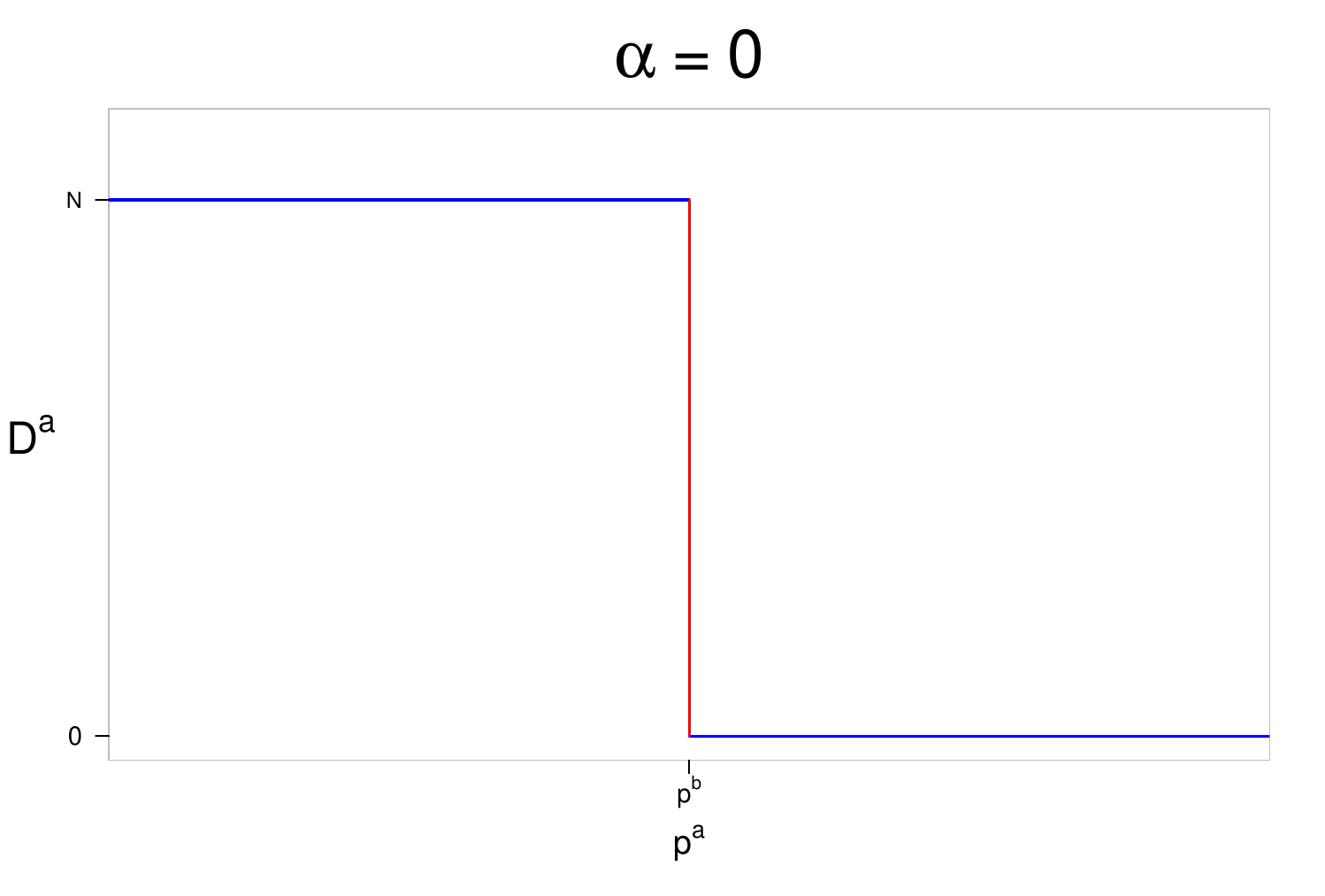}
        \includegraphics[width=0.67\textwidth]{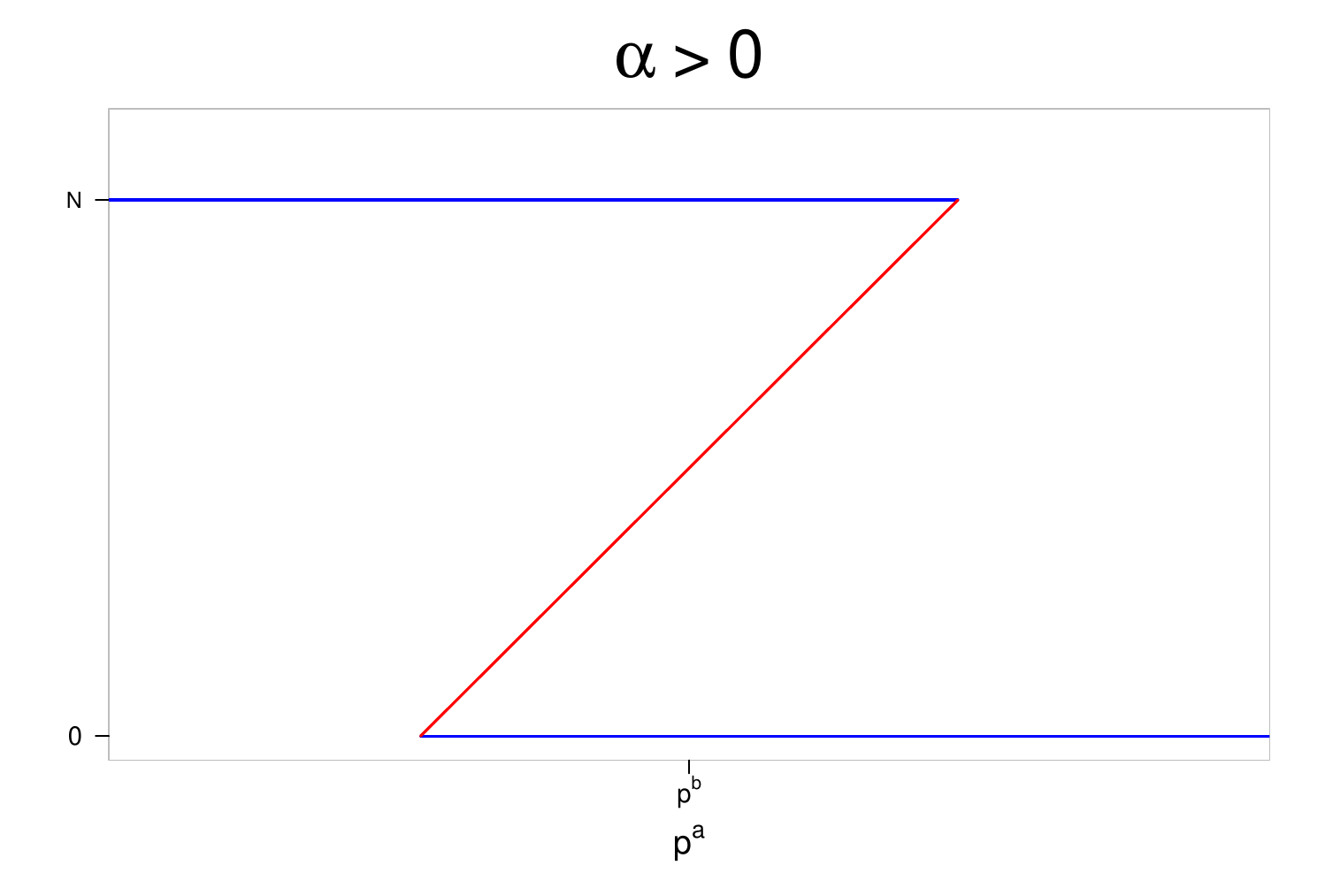}
        \caption{Depicted is a restriction of the equilibrium demand correspondence $Q$ to $\mathbb{R}_{\ge 0}\times\{p^b\}$, i. e. for a fixed price $p^b$ of firm $b$. On the left, the classical Bertrand duopoly model ($\alpha=0$). On the right, the case with positive network effects, $\alpha>0$. The vertical axis is $D^a$ while the horizontal is $p^a$.}
    \label{fig:A}
    \end{fullwidth}
\end{figure}

\newpage
In both above-discussed correspondences, no selection can induce a continuous demand function. In the Bertrand duopoly model, the demand function has a discontinuity at $p^a=p^b$. With network effects ($\alpha>0$), it is possible to make a selection that produces a demand function with only one discontinuity, as in the Bertrand case (although with the restriction $Q(p,p)=\{0,\frac{1}{2}, N\}$), or a translated version with the same form. However, it is also possible to have a demand function with as many jumps as desired. The problem with these two correspondences is that for any selection firms will be drawn towards the discontinuities, even if just one, because, when continuous, demand is either constant or upward-sloping. So, an equilibrium can only lie at a discontinuity, independently of the selection. As such, the demand response to any price deviation is \textit{abrupt}, which reveals the unstable nature of these price competition models.

\subsection{Hotelling's solution}
The motivation, or rationale, for Hotelling's approach, becomes clear in the first lines of the seminal work \citep{Hotelling1929}: \enquote{A profound difference in the nature of the stability of a competitive situation results from} [the fact that in \enquote{actual business}], \enquote{some buy from one seller, some from another, in spite of moderate price differences} [and the response to a price deviation of one firm] \enquote{will in general take place continuously rather than in the abrupt way which has tacitly been assumed.}

The idea proposed by Hotelling to produce continuity is to introduce heterogeneity in the function consumers are maximizing ($v(j)=U-p^j$), making it dependent on each consumer $i$ and firm $j$, by transforming the value $U$ into a function $U(i,j)$. The properties of the distribution of consumer characteristics that determine $U(i,j)$ then translate into properties of the demand function. \citet{Caplin:1991:AAI} show that a pure-strategy price equilibrium exists if the density satisfies a weak form of concavity and has convex support. Hotelling's \textit{sufficient heterogeneity} assumption is \textit{the} widespread solution to the problem and appears under different forms besides standard vertical or horizontal product differentiation. It includes consumer price sensitivities \citep{AllenThisse:1992}, large enough heterogeneity in consumers' tastes \citep{Palma:1985}, or increasing spread of household demand \citep{Hildenbrand:1994}, etc.

The argument for the stability of price competition put forward by Hotelling is inherently a local one: \enquote{If the purveyor of an article gradually increases his price while his rivals keep theirs fixed, the diminution in volume of his sales will in general take place continuously (...)}. It is, in fact, a continuity of demand in the neighborhood of equilibrium prices. Local continuity associated with $D^a=N-D^b$ will force differentiability, which, in turn, will require negative marginal demand in a neighborhood of equilibrium prices as a necessary condition for profitable outcomes. We call \textit{Hotelling selections} to selections from the equilibrium demand correspondence $Q$ that, given an outcome, have such properties.

\subsection{Limitations and network effects}
In the previously discussed models, there are no Hotelling selections. Positive network effects amplify the problem by folding the correspondence in the \textit{wrong} direction. In a model with positive network effects, the most known attempt to circumvent using Hotelling selections is probably Becker's \textit{note on restaurant pricing} \citep{Becker:1991}. The idea was to find an equilibrium where one firm has excess demand while the other has excess capacity. Firms, by fear of going from being 'in' to being 'out' would not deviate. After observing such a phenomenon while going out to dinner and disagreeing with his wife's assertion that the cause was a difference in the restaurant's amenities, Becker attempted to find a rationale using an upward-slopping demand as a primitive. \citet{Karni:1994:SAA} showed that such an equilibrium does not hold under Bertrand competition. As such, one byproduct would be that Becker's wife was right: the restaurants provided different amenities.

There are two caveats in building a network effects model on top of a Hotelling-type solution. One is the implied suggestion that shared market solutions exist because products are differentiated. Two is that network effects can disrupt the underlying properties induced in the demand function by a Hotelling-type assumption. So, its impact is generally assumed to be sufficiently small by imposing a bound. Network effects are a perturbation to the underlying model. The approach can be limiting if the object of study is network effects. Furthermore, it means the driver of consumers' decisions or outcomes is not network effects. Illustrative example \ref{ex:Grilo} discusses the issue through an explicit example. 

\begin{example}[From \cite{Grilo:2001:PCW}]
\label{ex:Grilo}
    Two stores located at $x^a$ and $x^b$, transportation cost $t$. The consumer located at $x$ choosing $a$ gets the following payoff (analogous for $b$),
        $$-p^a-t(x-x^a)^2+\alpha N\sigma.$$
    A shared equilibrium where both firms profit exists if, and only if,
        $$\alpha N < t(x^b-x^a).$$
    Let us freely interpret this inequality. The left-hand side is the maximum possible influence exerted on a consumer. The maximum value of network effects. On the right-hand side is the cost of changing the decision. So what the bound means is that even if all consumers exerted influence in the same direction to change someone's decision, that would never offset the cost of changing. The maximum network effects are lower than the cost of traveling from one firm to the other. From this perspective, it is a strong assumption to make if one wants to study network effects. The driver of the decision is the transportation cost. So, in some cases, we might want to remove the bound. In Figure \ref{f:ex:Grilo}, we depict how the effect translates into the respective correspondences, illustrating the effect on marginal demand.
\end{example}

\begin{figure}[ht]
    \begin{fullwidth}
        \includegraphics[width=0.29175\paperwidth]{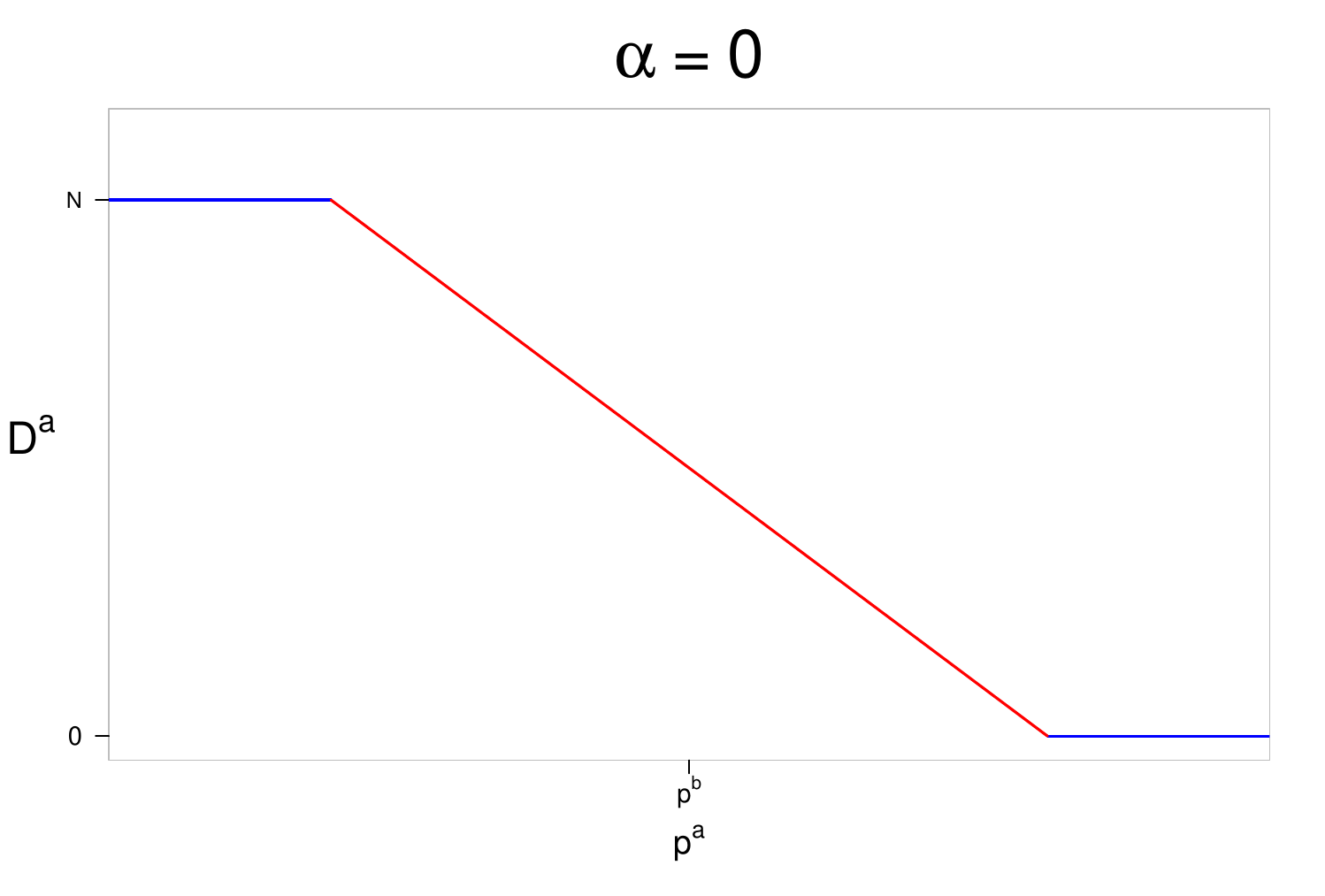}
        \includegraphics[width=0.29175\paperwidth]{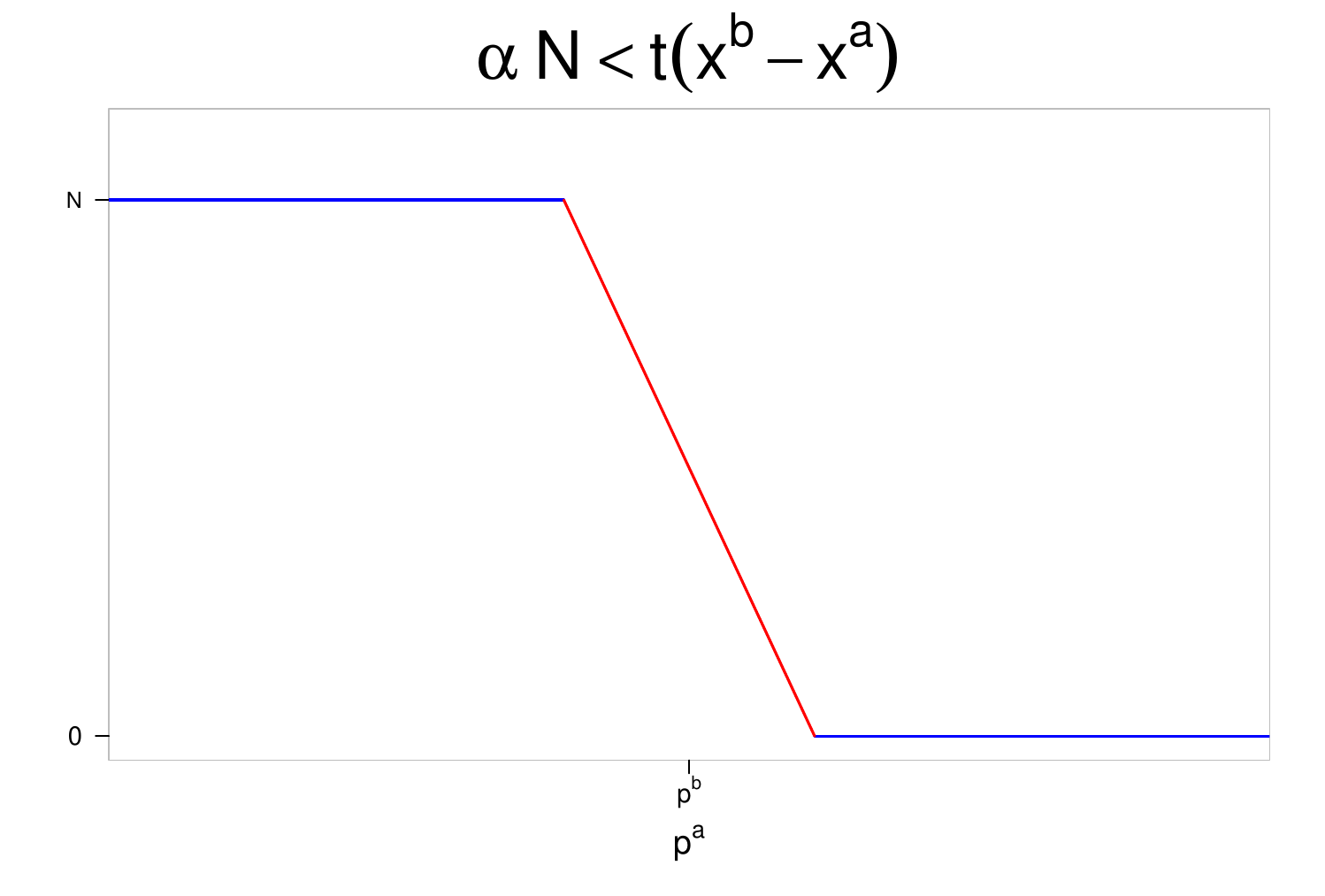}
        \includegraphics[width=0.29175\paperwidth]{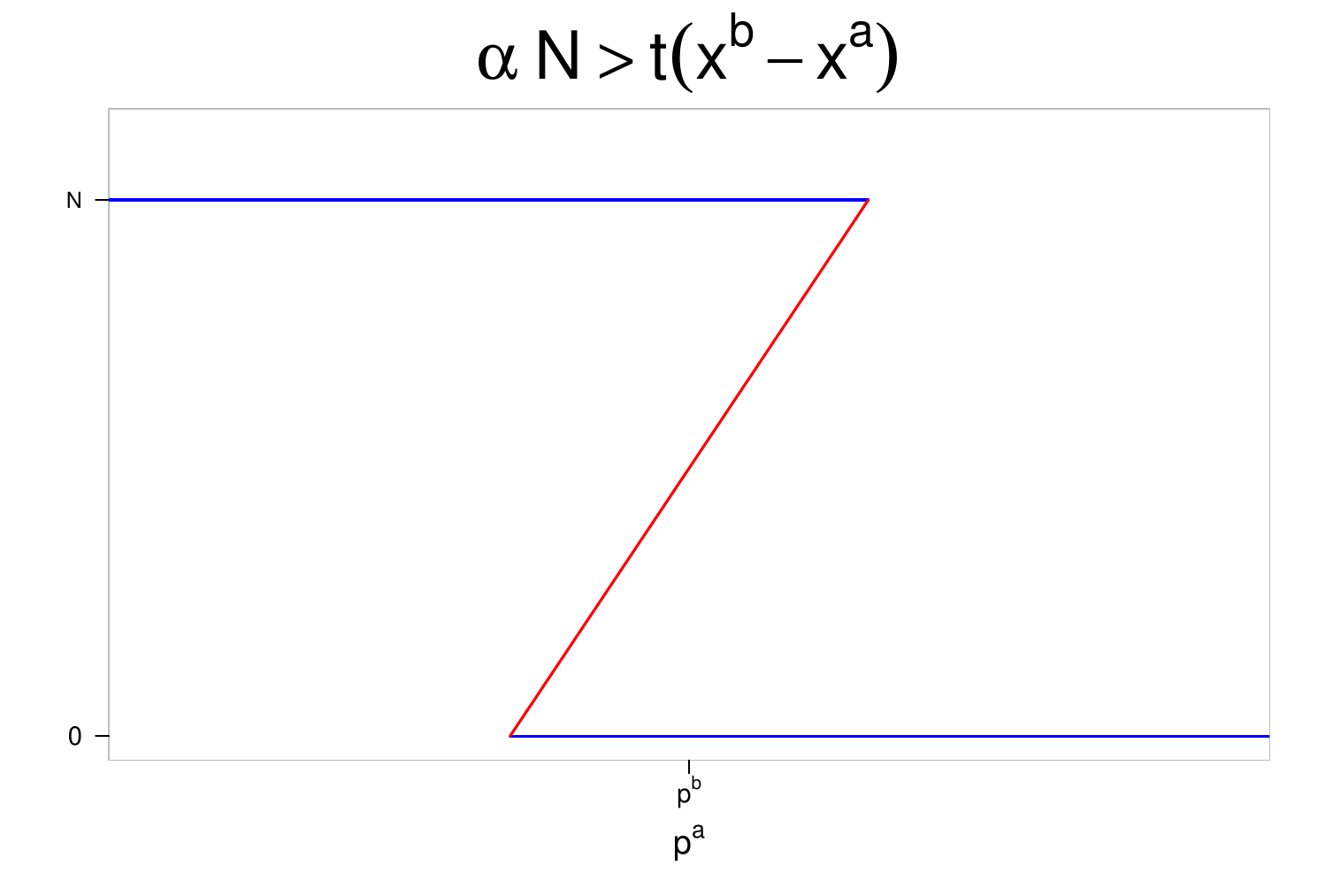}
    \caption{Correspondences of Illustrative example \ref{ex:Grilo} (restricted to $\mathbb{R}_{\ge 0}\times\{p^b\}$). From left to right we have the three models: no network effects; ($\alpha>0$) bounded positive network effects; ($\alpha>0$) a model that does not satisfy the bound. Note that the bound is needed to guarantee that the correspondence doesn't \textit{fold back}. The bound ensures that the \textit{stretch} from product differentiation overcomes the \textit{fold} from network effects.}
    \label{f:ex:Grilo}
    \end{fullwidth}
\end{figure} 

\subsection{A different approach}
In this note, we propose a different approach to create the possibility of Hotelling selections. We consider a partition of consumers into two groups, no bound on network effects, and no other factor, like product differentiation. We will characterize the interaction structures whose associated correspondences have Hotelling selections. Furthermore, we will show under which conditions an equilibrium where both firms have positive profits exists for a Hotelling selection.
The subsequent section has the setup and \nameref{s:model}. We end with a \nameref{s:discussion}.

\section{Main results}
\label{s:model}
\subsection{Partitioning consumers into groups}
Suppose there are two groups of consumers, indistinguishable within each group. For simplicity, we assume groups have the same size $N_1=N_2=1$, so $N=2$. Consumer choices are now represented by $\boldsymbol{\sigma}=(\sigma_1, \sigma_2)$, where $\sigma_i$ is the proportion of group $i$ who choose $a$ for the respective price. So $D^a=\sigma_1+\sigma_2$ and $D^b=2-(\sigma_1+\sigma_2)$. Assume $U(i,j)=0$. The interaction structure is given by
    $$
    M_{\alpha}= 
    \begin{pmatrix}
        \alpha_{11} & \alpha_{12}\\ 
        \alpha_{21} & \alpha_{22}
    \end{pmatrix}.
    $$
where $\alpha_{ii'}\ge 0$ represents how much an element of group $i$ is influenced by an element of group $i'$. For a member of group $i$, the value of choosing $a$ or $b$ under $\boldsymbol{\sigma}$ is given by
$$v_{i}(a;\boldsymbol{\sigma})=-p^a+\alpha_{i1} \sigma_1+\alpha_{i2}\sigma_2,$$
$$v_{i}(b;\boldsymbol{\sigma})=-p^b+\alpha_{i1}(1-\sigma_1)+\alpha_{i2}(1-\sigma_2).$$
The matrix $M_{\alpha}$ captures the underlying game structure, and we assume the determinant of $M_{\alpha}$ is non-zero. This new structure for the consumption stage, together with the pricing game in the first stage, forms a two-stage game, and we refer to it by $M_{\alpha}$. 

As everything relating to firms is symmetric, we will mostly focus on firm $a$. The equilibirum demand correspondence is now $Q:\mathbb{R}_{\ge0}^2\to [0,2]$ where $Q(\mathbf{p})$ is the set of values $2d$ for which there is a second stage equilibrium representable by $\boldsymbol{\sigma}$ and such that $\sigma_1+\sigma_2=d$. The notion of equilibrium is the standard subgame-perfect equilibrium, which we abbreviate to the equilibrium of a game $M_{\alpha}$. An equilibrium is $(\mathbf{p}^*,\boldsymbol{\sigma}(\mathbf{p}))$ where $\boldsymbol{\sigma}$ is a selection from $Q$ and $\mathbf{p}^*$ is a Nash equilibrium for the firms pricing game, with demand functions induced by $\boldsymbol{\sigma}(\mathbf{p})$.
\subsection{Results}
\label{ss:results}
We say that a group \textit{splits} when $0<\sigma_i<1$. We leave the second stage threshold characterization of the Nash equilibria for the appendix. It is a straightforward computation of the set of prices where each of the nine possible classes of equilibria\sidenote{There are nine classes because each group has three options, 0, 1 or to split, so nine combinations.} is selectable. Here, we will discuss only properties of the classes that induce nonhorizontal parts of the correspondence, i.e., equilibria where at least one group splits. Denote by $\boldsymbol{\sigma}|_P$ the restriction of the selection $\boldsymbol{\sigma}$ to a subset of prices $P$. 

Suppose $S$ is a subset of prices for which the equilibria where only group $i$ splits are always selectable. The marginal group demand for the selection $\boldsymbol{\sigma}|_S$ that always selects equilibria where only group $i$ splits, is
$$\frac{\partial \sigma_i}{\partial p^a}=\frac{1}{2\alpha_{ii}}=\frac{\partial (1-\sigma_i)}{\partial p^b}.$$
The selection that always chooses equilibria where both groups split is
$$\sigma_1(\mathbf{p})=\frac{1}{2}+\frac{{\kappa_1}}{2}(p^a-p^b), ~~~ \sigma_2(\mathbf{p})=\frac{1}{2}+\frac{\kappa_2}{2}(p^a-p^b),$$
with domain $|p^a-p^b| < \min\left\lbrace\frac{1}{|\kappa_1|},\frac{1}{|\kappa_2|}\right\rbrace$, and where
$$\kappa_1=\frac{\alpha_{22}-\alpha_{12}}{\det M_{\alpha}}, ~~~~\kappa_2=\frac{\alpha_{11}-\alpha_{21}}{\det M_{\alpha}}.$$
Hence,
$$\frac{\partial \sigma_i}{\partial p^a}=\kappa_i=\frac{\partial (1-\sigma_i)}{\partial p^b}.$$

\begin{lemma}
    \label{l:demand}
    A game $M_{\alpha}$ has a selection $\boldsymbol{\sigma}$ whose restriction to an open subset of prices $P$ is differentiable and induces negative marginal demand in $P$ if, and only if, the game has $\kappa_1+\kappa_2<0$ and $P$ is such that $|p^a-p^b| < \min\left\lbrace\frac{1}{|\kappa_1|},\frac{1}{|\kappa_2|}\right\rbrace$ for every $\mathbf{p}\in P$. 
    Furthermore, $\boldsymbol{\sigma}|_P$ represents the equilibria where both groups split, thus inducing, in $P$, the following demand $$D^j(\mathbf{p})=1+\frac{\kappa_1+\kappa_2}{2}(p^j-p^i).$$
\end{lemma}
The Lemma gives the conditions under which Hotelling selections exist in the correspondence of a given game $M_\alpha$. As $\frac{1}{2\alpha_{ii}}\ge 0$, the only equilibria which induce negative marginal demand are the ones where both groups split. Figure \ref{f:C} depicts three classes of correspondences for games that satisfy the conditions of Lemma \ref{l:demand}. The restriction $\boldsymbol{\sigma}|_P$ must be selected from the green line.
\begin{proof}
    $(\Rightarrow)$ Let $P\subset \mathbb{R}_{\ge0}^2$ and $\boldsymbol{\sigma}$ be a selection such that $\boldsymbol{\sigma}|_P$ is differentiable and induces negative marginal demand. As it is differentiable and $P$ open it must be that it keeps invariant whether a group splits or not (note that it is $\sigma_i$ which is differentiable and not $D^j$). As $\frac{\partial D^j}{\partial p^j}=\frac{\partial \sigma_1}{\partial p^j}+\frac{\partial \sigma_2}{\partial p^j}<0$, then the only possibility is to select the equilibria where both split, which is the only possibility for negative marginal demand, given by $\frac{\partial D^j}{\partial p^j}=\kappa_1+\kappa_2<0$. Hence, $P$ must be a subset of prices for which this class of equilibria exists.
    $(\Leftarrow)$ Simply take a $\boldsymbol{\sigma}|_P$ that always selects equilibria where both groups split. It is possible because $P$ is such that $|p^a-p^b| < \min\left\lbrace\frac{1}{|\kappa_1|},\frac{1}{|\kappa_2|}\right\rbrace$ for every $\mathbf{p}\in P$ (see Appendix \ref{app:tc}), and it is differentiable. It induces negative marginal because $\kappa_1+\kappa_2<0$.
\end{proof}

\begin{theorem}
    \label{t:prices}
    A pair of prices $\mathbf{p}$ is a local maximum for a selection $\boldsymbol{\sigma}$ which is continuous in a neighborhood $P$ of $\mathbf{p}$, and produces positive profits for both firms, if, and only if, $(i)$ $\kappa_1+\kappa_2<0$; $(ii)$ in $\boldsymbol{\sigma}|_P$ both groups split; and $(iii)$ prices are 
    $$p^{a}=p^{b}=p^*=\dfrac{2}{|\kappa_1+\kappa_2|}.$$
\end{theorem}
\begin{proof}
    $(\Leftarrow)$ With $(i)$ and $(ii)$ we get the expression for demand in $P$ given by Lemma \ref{l:demand}. Then, applying the standard FOC we get the solution $p^*$.
    $(\Rightarrow)$ Suppose $\boldsymbol{\sigma}$ is continuous in a neighbourhood $P$ of $\mathbf{p}$, and produces positive profits for both firms. As $D^a=N-D^b$, isoprofits are tangent, and because $p^j>0$ for $j=a,b$, demand must be differentiable in a neighborhood of $\mathbf{p}$, with negative derivative. Lemma \ref{l:demand} provides $(i)$, $(ii)$, and $(iii)$.
\end{proof}

By Theorem \ref{t:prices}, a local equilibrium with continuous $\boldsymbol{\sigma}(\mathbf{p})$ and positive profits only occurs when both groups split. By Lemma \ref{l:demand}, in that equilibrium each firm gets half of each group, leading to demand $D^*(\mathbf{p}^*)=1$. Each firm having one group (which would lead to the same $D=1$) is too stable for consumers, i.e., they would not react to small price changes. In the equilibria where both groups split, the response to price deviations is consumers from both groups move and change firms. When at most one group splits, a local equilibrium in the conditions of Theorem \ref{t:prices} does not exist because either demand is constant or non-differentiable, or marginal demand depends only on $\alpha_{ii}\ge 0$ if only group $i$ splits, thus being positive or constant.

\begin{proposition}
    \label{p:kpos}
    $\kappa_i>0$ for some group $i$.
\end{proposition}
\begin{proof}
	Suppose $\alpha_{22}<\alpha_{12}$. We have the following three cases:
	(i) 	if $\alpha_{11}<\alpha_{21}$, then $\det M_{\alpha}<0$ and $\kappa_1>0$ (and $\kappa_2 >0$);
	(ii) 	if $\alpha_{11}>\alpha_{21}$, then, if $\det M_{\alpha}<0$ we have $\kappa_1>0$ and $\kappa_2<0$; if $\det M_{\alpha}>0$ we have $\kappa_1<0$ and $\kappa_2>0$.
	(iii)	if $\alpha_{11}=\alpha_{21}=\alpha$, then $\kappa_1=\frac{1}{\alpha}>0$ and $\kappa_2=0$.
	The remaining cases ($\alpha_{22}\ge \alpha_{12}$) follow analogously.
\end{proof}
The marginal demand of each group has the same sign for both firms. From Proposition \ref{p:kpos}, we get that $\kappa_1+\kappa_2<0$ implies $\kappa_1\kappa_2<0$. Hence, it is positive for one group and negative for the other. In related literature on price competition with network effects and more than one group, the necessary condition of negative marginal demand to obtain pure price equilibria is often translated into negative marginal group demand \citep[e.g.][]{Banerji:2009:LNE}. It corresponds to the scenario where $\kappa_1<0$ and $\kappa_2<0$. By Propositon \ref{p:kpos}, the scenario is impossible in this setup. In a strategy where both groups split (or any other), it is impossible to have negative marginal consumer choices for both groups, so such an assumption can only lead to negative results (nonexistence).

\begin{remark}
    There are three classes of correspondences for games $M_{\alpha}$ with $\kappa_1+\kappa_2<0$. 
\end{remark}
The shape of the correspondence of a game depends only on the order of the thresholds of each class of equilibria.\sidenote{A careful analysis of Appendix \ref{app:tc} can make this precise.} Only those in which no group splits (there are four) can change order because when a group splits, the equilibria can only connect two other thresholds. The latter four classes are symmetric relative to $p^a=p^b$. The condition $\kappa_1+\kappa_2<0$ further restricts the possibilities ($\kappa_1>0$ induces an isomorphic correspondence to $\kappa_2>0$). Only two classes can change order (for example, for $\kappa_1>0$, the order of $(0,0)$ and $(1,0)$). When $\det M_\alpha>0$ and $\kappa_1+\kappa_2<0$, classes $(0,1)$ and $(1,0)$ are empty. Naturally, the sizes and slope of each segment may be different for correspondences in the same class.

\begin{figure}[ht]
    \begin{fullwidth}
        \includegraphics[width=0.2915\paperwidth]{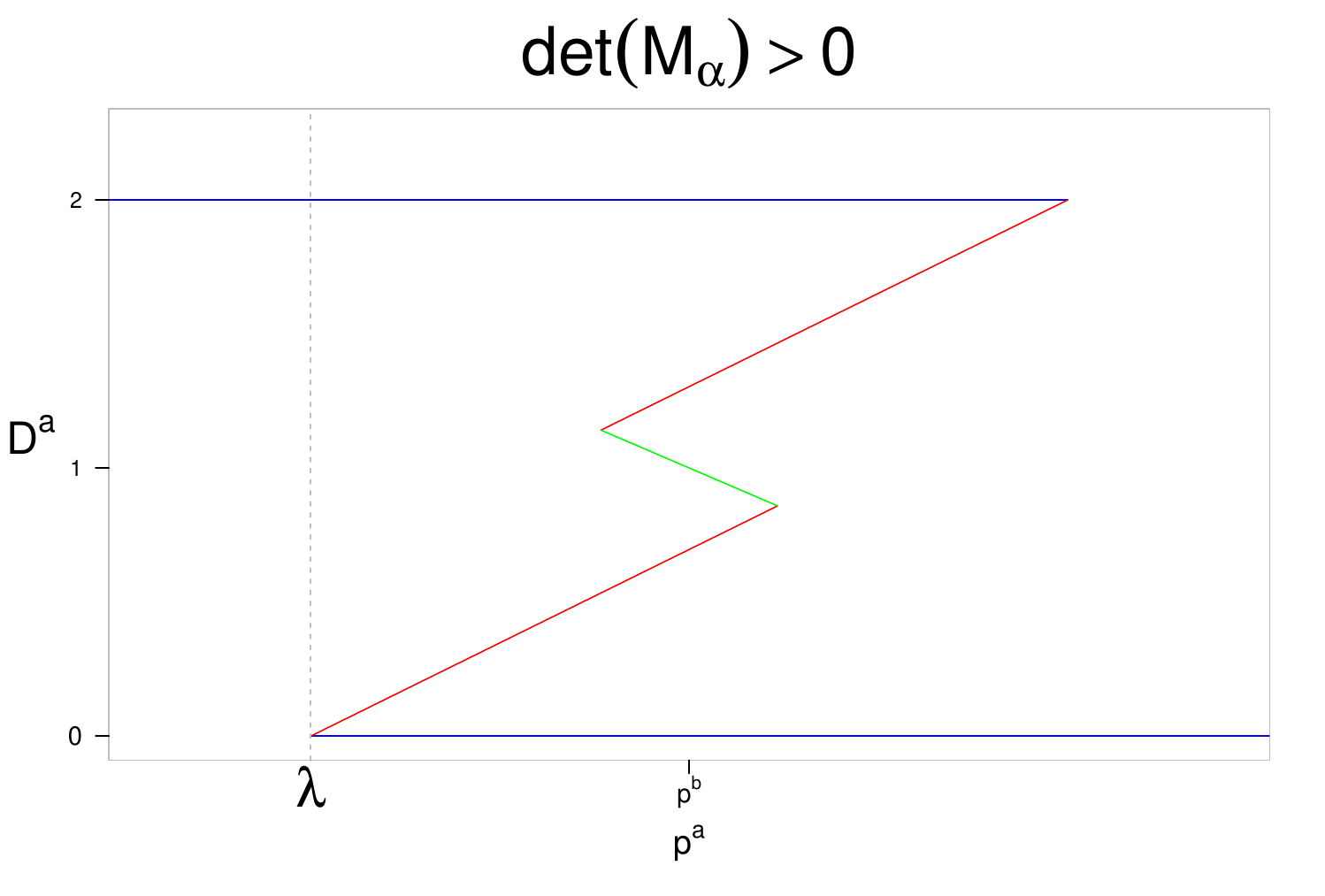}
        \includegraphics[width=0.2915\paperwidth]{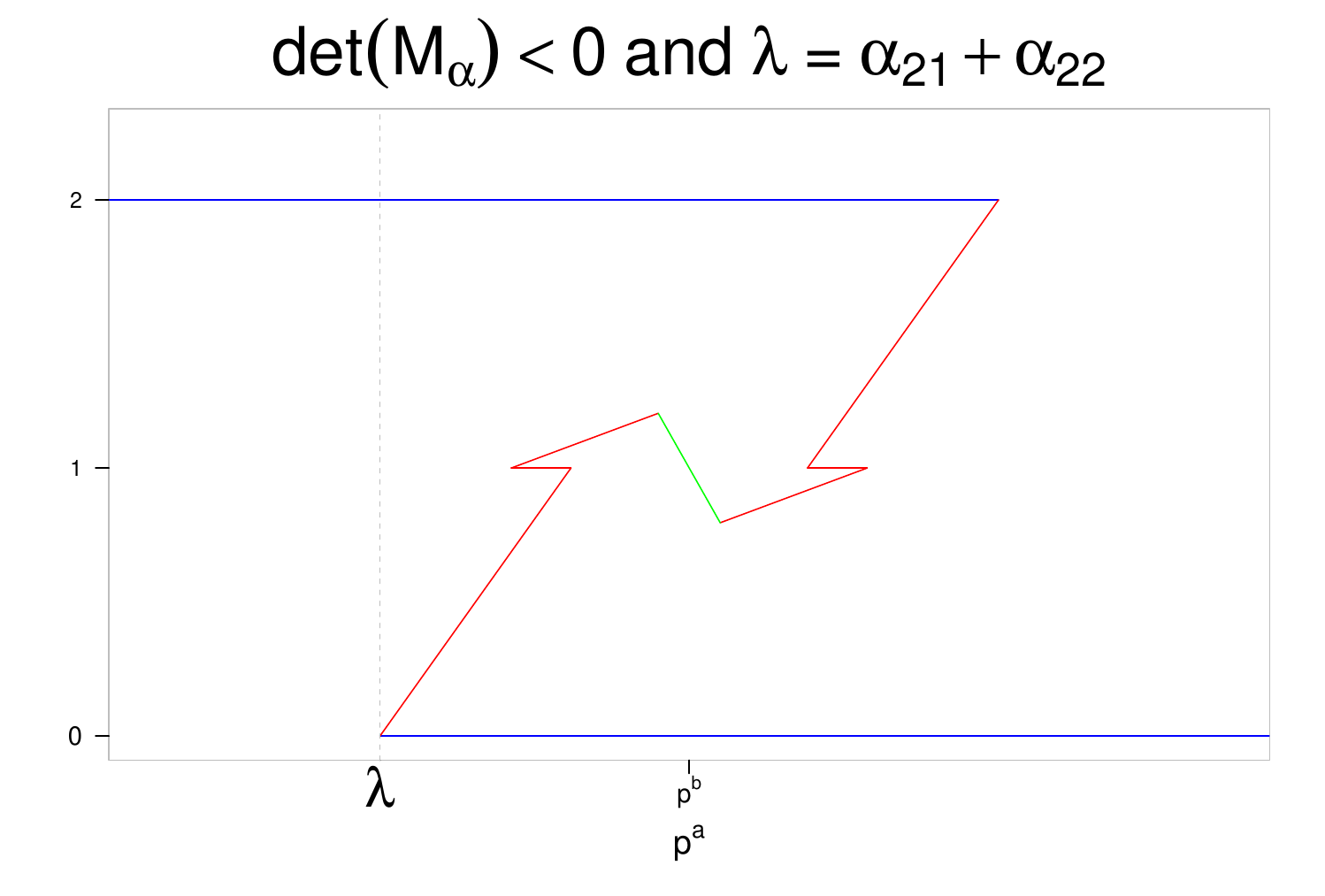}
        \includegraphics[width=0.2915\paperwidth]{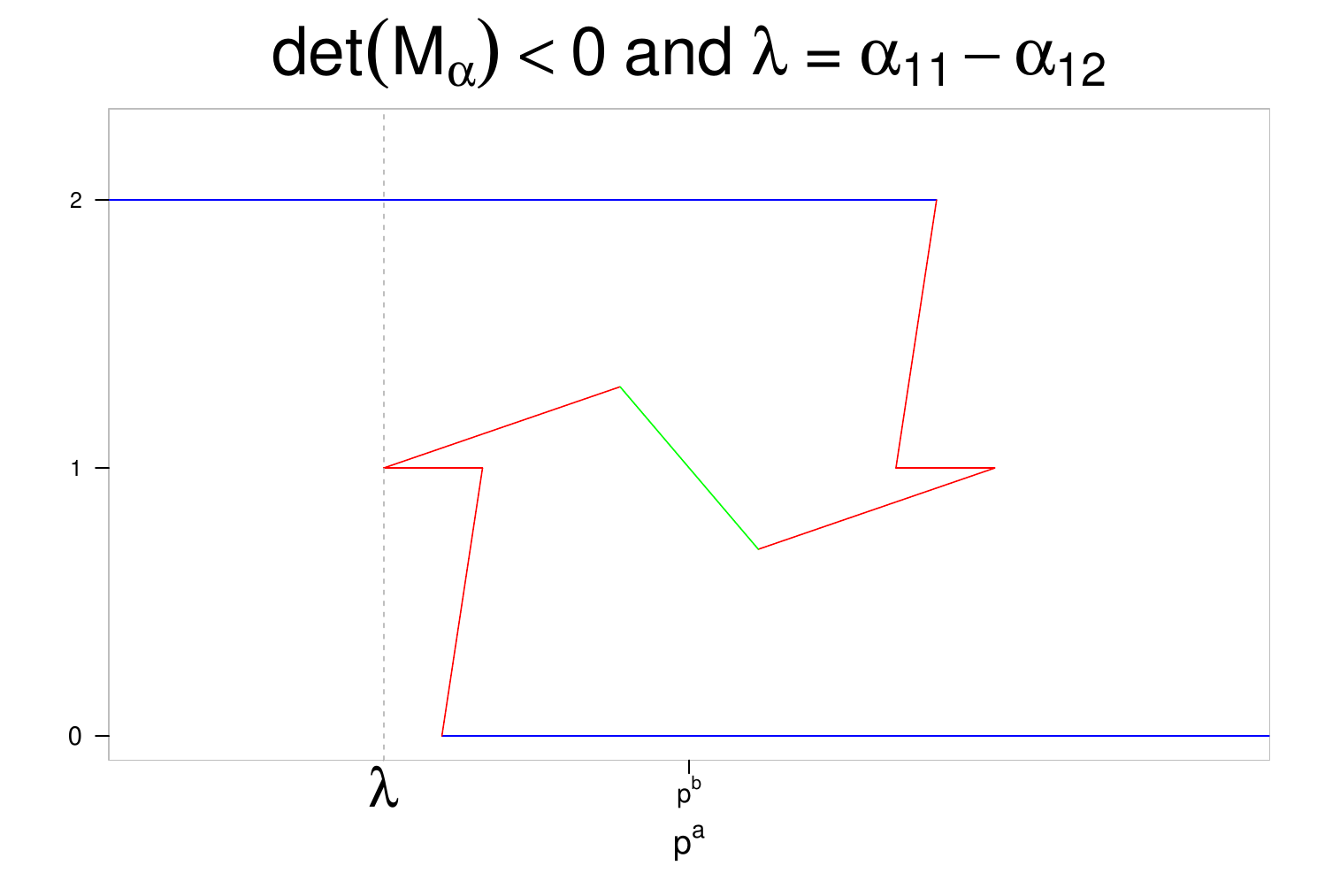}
    \caption{The three classes of correspondences (restricted to $\mathbb{R}_{\ge 0}\times\{p^b\}$) for games with $\kappa_1+\kappa_2<0$.}
    \label{f:C}
    \end{fullwidth}
\end{figure} 

We finish this section with the last result, a sufficient condition for a game to have an equilibrium with the outcome given by Theorem \ref{t:prices}. As $\kappa_1+\kappa_2<0$ forces $\kappa_1\kappa_2<0$ we will assume $\kappa_1>0$. Let\footnote{The definition of $\lambda$ is tighter than it needs to be. It means that if  $\kappa_1+\kappa_2> -\frac{1}{\lambda}$, no equilibrium contains the local equilibrium of Theorem \ref{t:prices}. A simpler choice for Theorem \ref{t:existence} would be, for example, $\lambda=\min_i\{\alpha_{i1}+\alpha_{i2}\}$, whose solution depends only on the sign of the determinant.}
$$\lambda=\lambda_{M_{\alpha},\kappa_1>0}=
	\left\lbrace 
	\begin{array}{ll}
		\alpha_{11}+\alpha_{12} & \mbox{if} ~~ \det M_{\alpha}>0 \\
		\max \{\alpha_{22}+\alpha_{21},\alpha_{11}-\alpha_{12}\} & \mbox{if} ~~  \det M_{\alpha}<0.\\ 
	\end{array}
    \right.
$$
Note that $\lambda>0$. Later, we will discuss an interpretation for $\lambda$. The role in the next theorem is that it depends on which of the three classes of correspondences a game with $\kappa_1+\kappa_2<0$ belongs to. These are depicted in Figure \ref{f:C}.
\begin{theorem}
    \label{t:existence}
    There is an equilibrium with positive profit for both firms if $\kappa_1+\kappa_2\le -\frac{1}{\lambda}$.
\end{theorem}

The condition to ensure existence is that the slope of demand cannot be too close to zero. We can also rewrite $\kappa_1+\kappa_2<-\frac{1}{\lambda}$ (when $\kappa_1>0$) as $\kappa_1<|\kappa_2|-\frac{1}{\lambda}$, which means, besides marginal demand of each group having different signs, they need to be sufficiently different. Note also that equilibrium prices are $\frac{2}{|\kappa_1+\kappa_2|}$, when $\kappa_1+\kappa_2$ is too small equilibrium prices will be high, giving margin for the firm to still have a sufficiently high price after the price undercut which guarantees monopoly (that price is $p^*-\lambda$), thus making it profitable.

The proof (left for Appendix \ref{app:pt2}) constructs an extension of the local equilibrium of Theorem \ref{t:prices} to a (global) equilibrium. We first show that the condition given in Theorem \ref{t:existence} ensures the local equilibrium of Theorem \ref{t:prices} exists and implicitly defines the equilibrium selection for a subset of prices $P$ around $\mathbf{p}^*$, where both groups split. The construction then follows by dividing the deviation price region $\mathbb{R}_{\ge 0}\times\{p^*\}$ in three: (i) prices above $p^*-\lambda$ that are outside of the interval defined by $P$ (two regions); and (ii) prices below $p^*-\lambda$. The conclusion follows by showing: 1) there are always selections that induce a demand function for which deviations in (i) are not profitable; and 2) capturing the whole market for prices in (ii) produces profits below the local equilibrium profit.

In proving Theorem \ref{t:existence}, we needed to define the equilibrium selection for a deviation region where both groups splitting is not selectable anymore. At those prices, demand is necessarily non-differentiable. Furthermore, it's a region with a multiplicity of equilibria, where, in particular, firms can lose or gain the whole market.
In this light, it is somewhat more natural to interpret $\boldsymbol{\sigma}$ as the firms' \textit{deviation demand forecasts}. Given an outcome $(\mathbf{p}^*, \boldsymbol{\sigma}^*)$, firm $a$ evaluates whether a deviation is profitable according to $\boldsymbol{\sigma}|_{\mathbb{R}_{\ge 0}\times\{p^{b*}\}}$, and it is in that sense that, for a fixed $p^b$, the restriction of $\boldsymbol{\sigma}$ is a firm $a$ forecast about the consumer response to a price deviation. The results of this work suggest that network effects create a local coordination mechanism for demand forecasts. Outside that, a bankruptcy scenario creates a fear of deviation in firms. In line with Becker's argument, firms won't deviate because the fear of going from being 'in' to being 'out' and losing their whole market outweighs the possibility of increasing profit.

\citet{Becker:1991} argument can be associated with the scenario $\kappa_1+\kappa_2>0$, where Theorem \ref{t:existence} does not apply. However, if the phenomenon Becker alluded to happens only partially, not affecting all groups of consumers, then Theorem \ref{t:prices} shows that it is possible to obtain positive profits for both firms. The bandwagon effect can provoke positive marginal demand in one group of consumers, with $\kappa_i>0$, but not for the other group, with $\kappa_{j}<0$. In this case, we translate Becker's argument from demand to the group level, and we can still have negative (overall) marginal demand. With this, it is possible to find an equilibrium with positive profits for both firms and thus avoid the nonexistence result of \citet{Karni:1994:SAA} for Becker's case.\sidenote{The comment refers only to the argument in \citep{Becker:1991}. Here, there are no capacities.} So, maybe Becker was right: it is possible the restaurants had similar amenities.
\section{Discussion}
\label{s:discussion}

% \subsection{Interpretation of $\kappa_1+\kappa_2<0$.}
    What characteristics do games with $\kappa_1+\kappa_2<0$ have? Are there many instances? Do they represent natural economic or social situations? The key to understanding this lies in the influence groups have on each other and whether it \textit{'flows'} in the same direction. Let us start with the following definition.
    \begin{definition}
        We say that the influence of group i is \textbf{centripetal}, if $\alpha_{ii}>\alpha_{ji}$, and \textbf{centrifugal}, if $\alpha_{ii}<\alpha_{ji}$.
    \end{definition}
    The definition relates intragroup and intergroup influences. If a consumer of a centripetal group changes her decision, that will have a higher impact on consumers from its own group than from the other group. The opposite happens for a group with centrifugal influence.
    The concepts of centripetal and centrifugal influence describe how a consumer impacts another consumer. It does not reveal how a consumer is influenced by one group or the other. Suppose group 1 is centripetal, i.e., $\alpha_{11}>\alpha_{21}$. The latter condition does not define $\alpha_{12}$, which determines how group 2 impacts group 1.

    The determinant reveals whether the group with negative marginal demand is centripetal or centrifugal. Note that $\kappa_{1}\det M_{\alpha}=\alpha_{22}-\alpha_{12}$. Suppose $\det M_{\alpha}>0$. If group $2$ has centripetal influence, then $\kappa_{1}>0$. If group $2$ has centrifugal influence, then $\kappa_{1}<0$. The reverse happens if $\det M_{\alpha}<0$. In sum, the type of influence of one group determines the sign of the marginal demand of the other group ($\kappa_i$).

    By Proposition \ref{p:kpos}, in a game where $\kappa_1+\kappa_2<0$, we must have $\kappa_1\kappa_2<0$, which leads to the following remark.
    \begin{remark}
        \label{r:centri}
        A game $M_{\alpha}$ has $\kappa_1\kappa_2<0$ if, and only if, the influence of one group is centripetal and the other centrifugal. If the influence of both groups is centripetal (or centrifugal), then $\kappa_1+\kappa_2>0$.
    \end{remark}
    The consequence of Remark \ref{r:centri} is that a game has Hotelling selections and the corresponding equilibria discussed in the previous sections only if one group has centripetal influence and the other has centrifugal influence.

    Let $cp$ and $cf$ denote the group with centripetal and centrifugal influence. We note that $(\kappa_1+\kappa_2)\det M_{\alpha}=\alpha_{11}+\alpha_{22}-(\alpha_{12}+\alpha_{21})$. When $\kappa_1+\kappa_2<0$, the following relation between inter and intragroup influences holds,
    \begin{enumerate}[label=(\roman*)]
        \item when $\det M_{\alpha}>0$ we have $\alpha_{11}+\alpha_{22}<\alpha_{12}+\alpha_{21}$ and $\kappa_{cp}<0$.
        \item when $\det M_{\alpha}<0$ we have $\alpha_{11}+\alpha_{22}>\alpha_{12}+\alpha_{21}$ and $\kappa_{cf}<0$
    \end{enumerate}
    The sign of the determinant reverses whether the centripetal or centrifugal group dominates marginal demand.

\subsection{Illustrative example 2 (Numeric)}
    Inspired in \citet{Armstrong:2006:CIT}, let us discuss a structure with only cross-group externalities. The scope of the work is distinct from the one here, where firms are two-sided platforms, but we find it elucidating to discuss the underlying structure and how to modify it to meet the conditions for the results here. The interaction structure is the following,
    $$
    M_{\alpha}= 
    \begin{pmatrix}
        0 & \alpha_{12}\\ 
        \alpha_{21} & 0
    \end{pmatrix}.
    $$
    A Hotelling specification is used for each group, coupled with a bound on network effects, to guarantee shared market solutions. The influence of both groups is centrifugal. Thus, differentiation of platforms is necessary!

    It is possible to achieve the shared market solutions discussed in the previous sections without having differentiated platforms. If we modify the structure by changing the diagonal elements, making one group centripetal, say $\alpha_{11}$ from $0$, to a value higher than $\alpha_{21}$. In choosing, we must meet the conditions that ensure a shared market solution given by Theorem \ref{t:existence}. In the light of the examples in \citep{Armstrong:2006:CIT}, where one group is consumers, say group 1, and the other are providers, this means consumers impact each other choices. If we make it centripetal, it means that influence is greater than the one on providers. For example, having more consumers can work as a signaling of the quality or safety of the platform; it could also be that it is more enjoyable when other consumers are using it, or it could represent the online rating associated with it, etc.

    For clarity, let's proceed with a numerical example. Suppose $\alpha_{12}=4$ and $\alpha_{21}=1$, and we want to add $\alpha_{11}$,
    $$
    M_{\alpha}= 
    \begin{pmatrix}
        \alpha_{11} & 4\\ 
        1 & 0
    \end{pmatrix}.
    $$
    We get $\det M_{\alpha}=-4$, $\kappa_1=1$, $\kappa_2=\frac{1-\alpha_{11}}{4}$. Hence, for $\kappa_2<0$ we need $\alpha_{11}>1$. Now, noting that $\kappa_1+\kappa_2=\frac{5-\alpha_{11}}{4}$, we need $\alpha_{11}>5$ to apply Lemma \ref{l:demand} and Theorem \ref{t:prices}. Finally, let's meet the condition in Theorem \ref{t:existence}. As $\lambda=\max\{1,\alpha_{11}-4\}=\alpha_{11}-4$, that condition becomes $-\alpha_{11}^2+9\alpha_{11}-16<0$, which is true for $\alpha_{11}<\frac{-\sqrt{17}+9}{2}$ or $\alpha_{11}>\frac{-\sqrt{17}+9}{2}>5$. Therefore, we meet all conditions by choosing $\alpha_{11}>\frac{-\sqrt{17}+9}{2}>5$. So, for example, if we choose $\alpha_{11}=7$, we get equilibrium prices $p^*=4$.

    The example might give the idea that we need group 1, the centripetal group, to value more consumers from its group than from the other group $\alpha_{11}>\alpha_{12}$, but that is only an artifact of having $\alpha_{22}=0$. Suppose we want a structure in which that doesn't happen, as
    $$
    M_{\alpha}= 
    \begin{pmatrix}
        2 & 4\\ 
        1 & \alpha_{22}
    \end{pmatrix}.
    $$
    In this case, following analogously, to obtain $\kappa_1+\kappa_2<0$ we need to choose $2<\alpha_{22}<3$. Then, as $\lambda=6$, for Theorem \ref{t:existence}, we need $\alpha_{22}<\frac{11}{4}$. As such, we can choose any value for $\alpha_{22}$ in the interval $(2,\frac{11}{4})$ to guarantee a shared solution of the previous sections. So, for example, if we choose $\alpha_{22}=\frac{10}{4}$ we get equilibrium prices $p^*=4$.

\newpage
\section{Concluding remark}
\blockquote[\cite{Amir:2018}]{(...) industrial organization deviated from standard microeconomics early on by postulating a downward-sloping demand as a primitive. (...)
Given this dichotomy between industrial organization and the broader microeconomics field, there is some theoretical motivation for developing results in oligopoly theory that do not rely on a downward-sloping demand as a primitive.}

We connect with this motivation by providing another way of deriving a decreasing demand function from maximizing a utility. Namely, one that depends only on positive network effects and prices. 

There is another way in which this new possibility departs from relying on downward-sloping demand as a primitive. Studying network effects need not be built on top of a downward-slopping demand founded on some other effect, while network effects appear as a perturbation. The present note provides an alternative way to reverse the approach. Take network effects as your primitive element, then perturb the resulting demand with another effect. In that way, we have a different road to understanding the role of network effects in shaping market outcomes and decisions.
\printbibliography
% DON'T EDIT. If "endfloat" option is enabled all floats appear before appendices
\if@endfloat\clearpage\processdelayedfloats\clearpage\fi 

%%%%%%%%%%%%%%%%%%%%%%%%%%%%%%%%%%%%%%%%%%%%%%%%%%%%%%%%%%%%
%%% SUPPLEMENTARY MATERIAL / APPENDICES
%%%%%%%%%%%%%%%%%%%%%%%%%%%%%%%%%%%%%%%%%%%%%%%%%%%%%%%%%%%%
%% Sadly, we can't use floats in the appendix boxes. So they don't "float", but use \captionof{figure}{...} and \captionof{table}{...} to get them properly caption.
\begin{appendix}

   \begin{appendixbox}\label{app:pt2}
      \section{Proof of Theorem \ref{t:existence}}
(\textit{There is an equilibrium with positive profit for both firms if $\kappa_1+\kappa_2\le -\frac{1}{\lambda}$.})

\subsection{Proof}
Suppose $\kappa_1+\kappa_2\le-\frac{1}{\lambda}$ and, wlg, $\kappa_1>0$. 

We will build the equilibrium $(\mathbf{p}^*, \boldsymbol{\sigma})$. Assume $\mathbf{p}^*=(p^*,p^*)$ and note that we need only define the selection $\boldsymbol{\sigma}$ for $\mathbb{R}_{\ge 0}\times\{p^*\}$ and $\{p^*\}\times\mathbb{R}_{\ge 0}$ as the rest is unattainable by unilateral deviation, thus being irrelevant in the Nash equilibrium condition (note a choice always exists by \citep{Nash:1951:NCG}). 

Let $Q^a_{p^*}$ be the restriction of $Q$ to $\mathbb{R}_{\ge 0}\times\{p^*\}$ and $Q^b_{p^*}$ the restriction of $Q$ to $\{p^*\}\times\mathbb{R}_{\ge 0}$ (the values of $\boldsymbol{\sigma}$ attainable by $a$ or $b$ through deviation when the other firm charges $p^*$). As everything is symmetric for firms, we will do the proof for firm $a$ and it follows analogously for $b$. From now on we will omit the firm superscript.

\paragraph{Local equilibrium: $p \in (p^*-\frac{1}{|\kappa_2|}, p^*+\frac{1}{|\kappa_2|})$}

First observe that because $\kappa_2<0$ and $\lambda>0$, $\kappa_1+\kappa_2<-\frac{1}{\lambda}$ implies $\kappa_1+\kappa_2<0$, and $\min\left\lbrace\frac{1}{|\kappa_1|},\frac{1}{|\kappa_2|}\right\rbrace =\frac{1}{|\kappa_2|}$. Consider now a subset $P\subseteq \mathbb{R}_{\ge0}^2$ such that for every $\mathbf{p}\in P$, $|p^a-p^b| < \frac{1}{|\kappa_2|}$. In $P$, let $\boldsymbol{\sigma}$ always select equilibria where both groups split (which is possible, as $P$ is by definition its domain, see Appendix \ref{app:pt2}). Let $p^*=\frac{2}{|\kappa_1+\kappa_2|}$. By Theorem \ref{t:prices} $(\mathbf{p}^*, \boldsymbol{\sigma}|_P)$ is a local equilibrium where both firms have positive profit. 

We need now only to define $\boldsymbol{\sigma}$ outside $P$, so that no firm wants to deviate to prices $p>0$ such that $p<p^*-\frac{1}{\kappa_2}$ or $p>p^*+\frac{1}{\kappa_2}$. 

\paragraph{Regions $p<p^*-\frac{1}{|\kappa_2|}$ and $p>p^*+\frac{1}{|\kappa_2|}$.}
We will divide the first region in two, and use the following claim.
\begin{claim} 
    \label{c:min}
    The following holds,
    \begin{enumerate}[label=\emph{(\roman*)}]
        \item $\min \{Q_{p^*}(p)\}\in\{0,1\}$ for every $p\in (p^*-\lambda, p^*-\frac{1}{|\kappa_2|})$;
        \item $\min \{Q_{p^*}(p)\}=0$ for every $p> p^*+\frac{1}{|\kappa_2|}$.
    \end{enumerate}
\end{claim}
Using Lemma \ref{l:demand}, we get that $D(\mathbf{p}^*)=1$. As such, let $\boldsymbol{\sigma}$ always select the minimum value of $Q_{p^*}(p)$ for $p$ in the regions (i) and (ii) of Claim \ref{c:min}. By the latter claim, we get that $D(p,p^*)\le 1$ for $(p^*-\lambda, p^*-\frac{1}{|\kappa_2|})$, and that $D(p,p^*)=0$ for $p> p^*+\frac{1}{|\kappa_2|}$. Therefore, the firm has no gain from deviating to those prices. We are left with the region $p\le p^*-\lambda$ to conclude the proof.
\vfill
\vspace*{1ex}
Note that $\Pi(p,p^*)\le 2p$, for all $p\ge0$. Let us show that $\Pi(\mathbf{p}^*)\ge 2 p$ for $p<p^*-\lambda$. Note that $\Pi(\mathbf{p}^*)=p^*$, thus, we need to show $p^*\ge 2(p^*-\lambda)$, that is $p^*\le 2\lambda$. As $p^*=\frac{2}{|\kappa_1+\kappa_2|}$, we get $\frac{1}{|\kappa_1+\kappa_2|}\le \lambda$ which is $\kappa_1+\kappa_2\le-\frac{1}{\lambda}$, true by assumption. 

%%-----------------------------------------------------------------PROOF OF CLAIM
\subsubsection{Proof of Claim}
Recall that
$$\lambda=\lambda_{M_{\alpha},\kappa_1>0}=
	\left\lbrace 
	\begin{array}{ll}
		\alpha_{11}+\alpha_{12} & \mbox{if} ~~ \det M_{\alpha}>0 \\
		\max \{\alpha_{22}+\alpha_{21},\alpha_{11}-\alpha_{12}\} & \mbox{if} ~~  \det M_{\alpha}<0.,\ 
	\end{array}
    \right.
$$
and$$\kappa_1=\frac{\alpha_{22}-\alpha_{12}}{\det M_{\alpha}}, ~~~~\kappa_2=\frac{\alpha_{11}-\alpha_{21}}{\det M_{\alpha}}.$$
When $\kappa_1>0$ and $\kappa_2<0$, we have, respectively, that, (1) $\det M_\alpha>0$ implies $\alpha_{22}>\alpha_{12}$ and $\alpha_{11}<\alpha_{21}$; and (2) $\det M_\alpha<0$ implies $\alpha_{22}<\alpha_{12}$ and $\alpha_{11}>\alpha_{21}$. As such, (1) if $\det M_\alpha>0$, then $\lambda=\alpha_{11}+\alpha_{12}$; and (2) if $\det M_\alpha<0$, then $\min_i\{\alpha_{i1}+\alpha_{i2}\}=\alpha_{22}+\alpha_{21}$. Thus, when $\alpha_{21}+\alpha_{22}>\alpha_{11}-\alpha_{12}$, we get $\lambda=\min_i\{\alpha_{i1}+\alpha_{i2}\}$.

% -------------------------------------------------------------------proof of (i)
\paragraph{(i) $\min \{Q_{p^*}(p)\}\in\{0,1\}$ for every $p\in (p^*-\lambda, p^*-\frac{1}{|\kappa_2|})$.}
By the threshold characterization in Appendix \ref{app:tc}, we have $0\in Q_{p^*}(p)$ for all $p>p^* -\min_i\{\alpha_{i1}+\alpha_{i2}\}$.
Suppose $\alpha_{21}+\alpha_{22}>\alpha_{11}-\alpha_{12}$, then $\lambda=\min_i\{\alpha_{i1}+\alpha_{i2}\}$. 
As such, in these cases, $0\in Q_{p^*}(p)$ for all $p>p^*-\lambda$.
Suppose now $\det M_{\alpha}<0$ and $\alpha_{21}+\alpha_{22}<\alpha_{11}-\alpha_{12}$. Note that $\min_i\{\alpha_{i1}+\alpha_{i2}\}=\alpha_{22}+\alpha_{21}$, now different from $\lambda=\alpha_{11}-\alpha_{12}$.
By the threshold characterization in Appendix \ref{app:tc}, $(1,0)$ represents an equilibrium if $p\in (p^*-(\alpha_{11}-\alpha_{12}),p^*+\alpha_{22}-\alpha_{21})$, thus for the prices in the latter interval $1\in Q_{p^*}(p)$. All is left to show is that, if $p^*+\alpha_{22}-\alpha_{21}<p^*-\frac{1}{|\kappa_2|}$ producing a non-empty interval $I$ where $(1,0)$ no longer has an underlying equilibrium, then for all $p\in I$, $\min \{Q_{p^*}(p)\}\in\{0,1\}$. But $p^*+\alpha_{22}-\alpha_{21}>p^*-(\alpha_{22}+\alpha_{21})=p^* -\min_i\{\alpha_{i1}+\alpha_{i2}\}$, which means that, if $I\neq\emptyset$, then $0\in Q_{p^*}(p)$, for all $p\in I$.

% -------------------------------------------------------------------proof of (ii)
\paragraph{(ii) $\min \{Q_{p^*}(p)\}=0$ for every $p> p^*+\frac{1}{|\kappa_2|}$.}
By the threshold characterization in Appendix \ref{app:tc}, we have $0\in Q_{p^*}(p)$ for all $p>p^* -\min_i\{\alpha_{i1}+\alpha_{i2}\}$. But $p^*+\frac{1}{|\kappa_2|}>p^* -\min_i\{\alpha_{i1}+\alpha_{i2}\}$ because $\min_i\{\alpha_{i1}+\alpha_{i2}\}>0$.

  \end{appendixbox}

   \begin{appendixbox}\label{app:tc}
       \section{Threshold characterization of consumers NE}
There are nine classes of consumer strategies.
We will characterize these classes. Let $\Delta v_i(\boldsymbol{\sigma})=\Delta v_i(\mathbf{p},\boldsymbol{\sigma})=v_i(a;\mathbf{p},\boldsymbol{\sigma})-v_i(b;\mathbf{p},\boldsymbol{\sigma})$.
Note that,
$$\Delta v_i(\boldsymbol{\sigma})= -\Delta p + \alpha_{i1}(2\sigma_1-1)+\alpha_{i2}(2\sigma_2-1).$$
Suppose $\boldsymbol{\sigma}$ represents a second stage Nash equilibrium. Observe that, 
\begin{itemize}
    \item $\sigma_i=1$ if, and only if, $\Delta v_i(\mathbf{p},\boldsymbol{\sigma})\ge 0$;
    \item $\sigma_i=0$ if, and only if, $\Delta v_i(\mathbf{p},\boldsymbol{\sigma})\le 0$;
    \item $0<\sigma_i<1$ if, and only if, $\Delta v_i(\mathbf{p},\boldsymbol{\sigma})=0$.
\end{itemize}
Let $E(\mathbf{p})$ be the set of $\boldsymbol{\sigma}$ that represent a second stage NE for $\mathbf{p}$. The following characterizes $E(\mathbf{p})$ through $\Delta p=p^a-p^b$.

\paragraph{No group splits.}
$(1,1)\in E(\mathbf{p})$ if, and only if, $\Delta p\le\min_i\{\alpha_{i1}+\alpha_{i2}\}$;\\
$(0,0)\in E(\mathbf{p})$ if, and only if, $\Delta p\ge-\min_i\{\alpha_{i1}+\alpha_{i2}\}$.\\
$(1,0)\in E(\mathbf{p})$ if, and only if, $-(\alpha_{22}-\alpha_{21})\le \Delta p \le \alpha_{11}-\alpha_{12}$;\\
$(0,1)\in E(\mathbf{p})$ if, and only if, $-(\alpha_{11}-\alpha_{12})\le \Delta p \le \alpha_{22}-\alpha_{21}$.\\

% (q,0)
\paragraph{Only group 1 splits ($0<\sigma_1<1$) - case $(\sigma_1,0)$.}
We need to satisfy $\Delta p = \alpha_{11}(2\sigma_1-1)-\alpha_{12}$, and $\Delta p \ge \alpha_{21}(2\sigma_1-1)-\alpha_{22}$. 
The first has an exact solution and a domain, while the second (substituting by the solution obtained) leads to $\Delta p(\alpha_{11}-\alpha_{21})>- \det M_{\alpha}$. Therefore, noting that $\kappa_2=\frac{\alpha_{11}-\alpha_{21}}{\det M_{\alpha}}$, we get the following cases.\\

When $\alpha_{11}\neq 0$, then $(\sigma_1,0)\in E(\mathbf{p})$ if, and only if,
$$\sigma_1=\frac{\Delta p+\alpha_{11}+\alpha_{12}}{2\alpha_{11}},$$
and $\Delta p$ satisfies
\begin{enumerate}[label=(\roman*)]
    \item $\max\left\{-\frac{1}{\kappa_2},-(\alpha_{11}+\alpha_{12})\right\}<\Delta p<\alpha_{11}-\alpha_{12}$ when $\alpha_{11}>\alpha_{21}$;
    \item $-(\alpha_{11}+\alpha_{12})<\Delta p<\min\left\{\alpha_{11}-\alpha_{12},-\frac{1}{\kappa_2}\right\}$ when $\alpha_{11}<\alpha_{21}$;
    \item $-(\alpha_{11}+\alpha_{12})<\Delta p<\alpha_{11}-\alpha_{12}$ when $\alpha_{11}=\alpha_{21}$, if $\det M_{\alpha}>0$.
\end{enumerate}

When $\alpha_{11}=0$, then $(\sigma_1,0)\in E(\mathbf{p})$ if $\Delta p=-\alpha_{12}$ and $\sigma_1$ satisfies
$$0<\sigma_1 < \frac{\alpha_{22}+\alpha_{21}-\alpha_{12}}{2\alpha_{21}}.$$

%(q,1)
\paragraph{Only group 1 splits ($0<\sigma_1<1$) - case $(\sigma_1,1)$}
We need to satisfy $\Delta p = \alpha_{11}(2\sigma_1-1)-\alpha_{12}$, and $\Delta p \le \alpha_{21}(2\sigma_1-1)+\alpha_{22}$. 
As before, this leads to $\Delta p(\alpha_{11}-\alpha_{21})< \det M_{\alpha}$. Therefore, noting that $\kappa_2=\frac{\alpha_{11}-\alpha_{21}}{\det M_{\alpha}}$, we get the following cases.\\

When $\alpha_{11}\neq 0$, then $(\sigma_1,1)\in E(\mathbf{p})$ if, and only if,
$$\sigma_1=\frac{\Delta p+\alpha_{11}-\alpha_{12}}{2\alpha_{11}},$$
and $\Delta p$ satisfies
\begin{enumerate}[label=(\roman*)]
    \item $-(\alpha_{11}-\alpha_{12})<\Delta p<\min\left\{\alpha_{11}+\alpha_{12},\frac{1}{\kappa_2}\right\}$ when $\alpha_{11}>\alpha_{21}$;
    \item $\max\left\{-(\alpha_{11}-\alpha_{12}),\frac{1}{\kappa_2}\right\}<\Delta p<\alpha_{11}+\alpha_{12}$ when $\alpha_{11}<\alpha_{21}$;
    \item $-(\alpha_{11}-\alpha_{12})<\Delta p<\alpha_{11}+\alpha_{12}$ when $\alpha_{11}=\alpha_{21}$, if $\det M_{\alpha}>0$.
\end{enumerate}

When $\alpha_{11}=0$, then $(\sigma_1,0)\in E(\mathbf{p})$ if, and only if, $\Delta p=\alpha_{12}$ and $\sigma_1$ satisfies
$$0<\sigma_1 < \frac{\alpha_{12}+\alpha_{21}-\alpha_{22}}{2\alpha_{21}}.$$

\paragraph{Only group 2 splits ($0<\sigma_2<1$)}
Cases $(0,\sigma_2)$ and $(1,\sigma_2)$ are analogous, just reverse the role of 1 and 2.

\paragraph{Both groups split}
For $0<\sigma_i<1$, $(\sigma_1, \sigma_2)\in E(\mathbf{p})$ if, and only if, it is a solution to the system $\Delta p = \alpha_{11}(2\sigma_1-1)+\alpha_{12}(2\sigma_2-1)$, and $\Delta p = \alpha_{21}(2\sigma_1-1)+\alpha_{22}(2\sigma_2-1)$. The solution is
$$\sigma_1=\frac{\kappa_1}{2}\Delta p+\frac{1}{2}, ~~ \text{and} ~~ \sigma_2=\frac{\kappa_2}{2}\Delta p+\frac{1}{2},$$
which satisfies $0<\sigma_i<1$ when $-1<\kappa_i\Delta p<1$ for $i=1,2$. Thus, when $\Delta p$ is such that
$$-\min_i\left\{\frac{1}{|\kappa_i|}\right\}<\Delta p<\min_i\left\{\frac{1}{|\kappa_i|}\right\}.$$

\subsection{Empty classes}
There are empty classes for games with $\det M_\alpha>0$ and $\kappa_1+\kappa_2<0$. In the latter games, the size of the price domain for the classes $(1,0)$ and $(0,1)$ is $\alpha_{11}+\alpha_{22}-\alpha_{12}+\alpha_{21}$ which is $(\kappa_1+\kappa_2)\det M_\alpha<0$. Furthermore, suppose, the same games have $\kappa_1>0$. Then $\alpha_{22}<\alpha_{12}$ (because of $\det M_\alpha>0$), which leads to $\max\left\{-(\alpha_{22}-\alpha_{21}),\frac{1}{\kappa_1}\right\}=\frac{1}{\kappa_1}$. The classes $(0,\sigma_2)$ and $(1,\sigma_2)$ have a price domain (interval) of size $\alpha_{22}-\alpha_{21}+\frac{1}{\kappa_1}$ which is also $(\kappa_1+\kappa_2)\det M_\alpha<0$.
   \end{appendixbox}

\end{appendix}
%%%%%%%%%%%%%%%%%%%%%%%%%%%%%%%%%%%%%%%%%%%%%%%%%%%%%%%%%%%%
%%% ARTICLE END
%%%%%%%%%%%%%%%%%%%%%%%%%%%%%%%%%%%%%%%%%%%%%%%%%%%%%%%%%%%%

\end{document}